\documentclass[conference]{IEEEtran}[10pt]

\usepackage{amsmath}
\usepackage{amsfonts}
\usepackage{amssymb}
\usepackage{graphicx}
\usepackage{psfrag}
\usepackage{cite}
\usepackage{algorithm}
\usepackage{color}
\usepackage{float}
\usepackage{url}
\usepackage{epsfig,array,multicol,verbatim,algorithmic}%
\IEEEoverridecommandlockouts
\newtheorem{theorem}{\bf Theorem}

\newtheorem{property}{\bf Property}
\newtheorem{proposition}{Proposition}
\newtheorem{definition}{\bf Definition}

\newtheorem{remark}{\bf Remark}
\newlength{\aligntop}
\newlength{\alignbot}
\makeatletter
\renewenvironment{align}{%
  \vspace{\aligntop}
  \start@align\@ne\st@rredfalse\m@ne
}{%
  \math@cr \black@\totwidth@
  \egroup
  \ifingather@
    \restorealignstate@
    \egroup
    \nonumber
    \ifnum0=`{\fi\iffalse}\fi
  \else
    $$%
  \fi
  \ignorespacesafterend%
  \vspace{\alignbot}\par\noindent
}
\newcommand{\mysection}[1]{\vspace*{-0.47em}\section{#1}\vspace*{-0.49em}}

\IEEEoverridecommandlockouts
\begin{document}

\title{{\huge Coalition Formation Games for Distributed  Cooperation Among Roadside Units in Vehicular Networks}}
\author{\authorblockN{Walid Saad$^\textbf{1}$, Zhu Han$^\textbf{2}$,  Are Hj{\o}rungnes$^\textbf{1}$, Dusit Niyato$^\textbf{3}$, and Ekram Hossain$^\textbf{4}$} \authorblockA{\small
$^\textbf{1}$UNIK - University Graduate Center, University of Oslo, Kjeller, Norway, email: \url{{saad, arehj}@unik.no}\\
$^\textbf{2}$ Electrical and Computer Engineering Department, University of Houston, Houston, USA, email: \url{zhan2@mail.uh.edu}\\
$^\textbf{3}$ School of Computer Engineering, Nanyang Technological University (NTU), Singapore, \url{dniyato@ntu.edu.sg}\\
$^\textbf{4}$ Department of Electrical and Computer Engineering at University of Manitoba, Canada , \url{ekram@ee.umanitoba.ca}\vspace{-1.05cm}
 }%
   \thanks{This work was supported by the Research Council of Norway through the project 183311/S10 and NSF grants CNS-0953377, CNS-0905556, and CNS-0910461. The work of E. Hossain was supported by the AUTO21 NCE research grant for the project F303-FVT.}}
\date{}
\maketitle

\begin{abstract}Vehicle-to-roadside (V2R) communications enable vehicular networks to support a wide range of applications for enhancing the efficiency of road transportation. While existing work focused on non-cooperative techniques for V2R communications between vehicles and roadside units (RSUs), this paper investigates novel cooperative strategies among the RSUs in a vehicular network. We propose a scheme whereby, through cooperation, the RSUs in a vehicular network can coordinate the classes of data being transmitted through V2R communications links to the vehicles. This scheme improves the diversity of the information circulating in the network while exploiting the underlying content-sharing vehicle-to-vehicle communication network. We model the problem as a coalition formation game with transferable utility and we propose an algorithm for forming coalitions among the RSUs. For coalition formation, each RSU can take an individual decision to join or leave a coalition, depending on its utility which accounts for the generated revenues and the costs for coalition coordination. We show that the RSUs can self-organize into a Nash-stable partition and adapt this partition to environmental changes. Simulation results show that, depending on different scenarios, coalition formation presents a performance improvement, in terms of the average payoff per RSU, ranging between $20.5\%$ and $33.2\%$, relative to the non-cooperative case.
\end{abstract}
\vspace{-0.1cm}
\mysection{Introduction}
Recent advances in the integration of communication and sensor technologies have triggered the deployment of numerous attractive applications for road transportation systems. In this regard, networks of connected vehicles constitute the main building block of intelligent transportation systems~(ITS) and are the basis for a diversity of applications that can enhance the safety and comfort of road transportation (e.g., through providing road traffic condition, remote vehicle monitoring, accident prevention, payment services, security applications etc.) \cite{VAPP}. In order to support different ITS applications, both vehicle-to-roadside~(V2R) communications and vehicle-to-vehicle~(V2V) communications need to be supported in vehicular networks. On one hand, V2R communications allow the vehicles to connect, through their on-board units~(OBUs), to the roadside units~(RSUs) belonging to one or several service providers, in order to download (or upload) various types of data related to a variety of applications. On the other hand, V2V communications enable a group of vehicles to communicate and exchange information for different purposes.

Existing work has already explored various aspects of V2R and V2V communications. For instance, in \cite{V2R00}, a low-complexity scheme for packet scheduling for downlink and uplink transmissions over V2R communication links (between an RSU and multiple OBUs) is proposed. The authors in \cite{V2R01} propose an IEEE 802.16-based protocol for data communication between a cluster of vehicles and an RSU. Further, an experimental testbed for traffic congestion detection and emergency warning using V2R and V2V communication is presented in \cite{V2V00}. In \cite{V2V01}, the authors propose a non-cooperative Bit Torrent-based approach for data distribution between the RSUs and the OBUs of the vehicles as well as a Nash bargaining solution for V2V data exchange. Multiple antenna techniques are proposed in \cite{V2R02} for enhancing the performance of V2R communications. The work in \cite{V2V02} studies a protocol, using time sharing for inter-vehicle message delivery with short and deterministic delay bounds in a V2V ad hoc network. The objective of the work in \cite{V2V02} is to improve the safety of vehicular networks. In addition, the work in \cite{V2V04} proposes an effective protocol, comprising congestion control policies, mechanisms for service differentiation, and methods for emergency warning dissemination using V2V communications. Further, the use of V2V communications and cooperation among sensor-equipped vehicles is studied in \cite{V2V03} for proactive urban data monitoring. Other aspects of V2R and V2V communication such as routing, security, channel modeling, and authentication are studied in \cite{OTHER00,OTHER01,OTHER02,OTHER03,NEW01,NEW02,NEW03}.

Most  of the existing work in vehicular networks have focused on communication technologies for V2R or V2V communication, content-sharing through V2V cooperation as well as on non-cooperative data delivery between the RSUs and the OBUs of the vehicles through V2R communications. Nonetheless, one challenging aspect of vehicular networks that remains unexplored is the design of cooperative strategies among \emph{the RSUs} for improving the diversity of the data circulating in the network as well as for exploiting the data exchange capabilities of the underlying V2V networks.  By exploring the possibilities of content-sharing between the vehicles, the RSUs in a vehicular network can cooperate in order to coordinate the classes of data that they will transmit to their served vehicles. For example, instead of non-cooperartively sending information on the traffic of the same geographical location to their served vehicles, two RSUs can cooperate to send information on the traffic conditions at different locations, and, subsequently, rely on the V2V data exchange for disseminating this data to all the vehicles traveling between them. Therefore, by using an efficient V2V data exchange protocol, all the vehicles moving between the two RSUs will acquire traffic information on different geographical areas without the need for passing by multiple RSUs (for example). By doing so, the RSUs can obtain more revenues from the vehicles since they are providing them with a more diverse amount of information through cooperation. Further, from the vehicles' perspective, due to the short duration that a vehicle spends at an RSU \cite{VAPP,V2R00,V2R01,V2V01}, it is common that a vehicle merely has time to download a limited number of chunks or packets, e.g., related to a single class of data. By enabling cooperation among the RSUs, the vehicles will be able to obtain more diverse classes of data,  by, for example, downloading a class from one RSU and engaging in V2V communication with other vehicles that obtained a different class of data from other cooperative RSUs. To the best of our knowledge, no existing work has studied this cooperation problem among RSUs, notably from a game theoretical perspective.

The main contribution of this paper is  a novel cooperation protocol that enables the RSUs in a vehicular network to maximize the revenues they obtain from the data they convey to their served vehicles. Through cooperation, the RSUs can diversify the classes of data that they transmit to their served vehicles, depending on the content-sharing possibilities of the underlying V2V network connecting the vehicles circulating between them.  We model the problem as a coalition formation game among the RSUs, and we propose a distributed algorithm for forming the coalitions. Through the proposed algorithm, each RSU can take a distributed decision to leave its current coalition and join a new one while maximizing its utility which accounts for the gains in terms of the total revenue generated from the data transmitted to the vehicles as well as the costs for coordination inside the coalition. We show that, by using the proposed coalition formation algorithm, the RSUs can self-organize into a Nash-stable partition, as well as adapt this partition to environmental changes such as a change in the average vehicle traffic passing by each RSU. Simulation results show that, depending on the scenario, coalition formation allows the RSUs to self-organize while improving their average payoff between $20.5\%$ and $33.2\%$, relative to the non-cooperative case.

The remainder of this paper is organized as follows: Section~\ref{sec:sysmodel} presents the proposed system model. In Section~\ref{sec:gmodel}, we model the problem of cooperation among RSUs as a transferable utility coalitional game and propose a suitable utility function. In Section~\ref{sec:gform}, we classify the proposed coalitional game as a coalition formation game, we discuss its key properties, and we introduce the algorithm for coalition formation. Simulation results are presented in Section~\ref{sec:sim}. Finally, conclusions are drawn in Section~\ref{sec:conc}.
\mysection{System Model}\label{sec:sysmodel}

Consider a network consisting of $N$ RSUs and let $\mathcal{N}$ denote the set of all RSUs. We denote by $\mathcal{C}=\{c_1,\ldots,c_L\}$ the set of cardinality $|\mathcal{C}| = L$ which represents the classes of data that can be distributed by any RSU $i \in \mathcal{N}$. Further, each RSU $i \in \mathcal{N}$ engages in a V2R communication with an average of $K_{ij}$ vehicles that enter into the network, pass by this RSU, and move towards RSU $j\in \mathcal{N},\ i \neq j$. We let $\mathcal{K}_{ij}$ denote the set containing $K_{ij}$ vehicles moving from any RSU $i$ towards any RSU $j$. Each vehicle $k \in \mathcal{K}_{ij}$ passing by an RSU $i \in \mathcal{N}$, connects to this RSU for a period of time and downloads an average of $P_{k,i}$ chunks of data. The data downloaded by vehicle $k$ belongs to class  $c_l^{ij} \in \mathcal{C}$ transmitted by RSU $i$ to all vehicles traveling in the direction of RSU $j$. Due to the short period that a vehicle can spend at an RSU, we consider that each RSU selects only \emph{one} class of data from $\mathcal{C}$ to transmit, at a time, to the vehicles in a given direction. Further, for multiple access, any scheme can be adopted by the vehicles as long as it ensures that they are able to successfully download the required data from the RSUs.

A payment operator, i.e., an entity which takes care of collecting payments on behalf of the RSUs, charges every vehicle $k$, after passing by its \emph{first} RSU and prior to meeting its next RSU (or exiting the network), a fee proportional to the total amount of data that this vehicle carries at that time. This amount of money is then distributed by the operator to the RSU(s) that transmitted the data obtained by vehicle $k$. For each class of data $c_k$, there is a corresponding priority (weight) $w_{c_k} \le 1$ which quantifies the importance of this data for the vehicular network. As the weight of the data increases, the operator charges the vehicles a higher price for receiving this data. For convenience, we consider that the data classes in the set $\mathcal{C}$ are ordered such that $w_{c_1} > w_{c_2} > \ldots > w_{c_L}$. Note that, although in some situations, the RSUs may also need to pay some stipulated fees for the underlying vehicular network, throughout this paper, we consider that this fee is a one time payment (for example, collected periodically every month or year) that will not affect the continuous revenue stream that the RSUs receive through the packets downloaded by the vehicles. Nonetheless, in scenarios where the fee paid by the RSUs is quite considerable, one can still apply the  approach proposed in the remainder of this paper.

In the considered network, we assume the presence of a V2V content-sharing scheme which allows the vehicles to communicate and exchange data among each other when possible.  However, in a non-cooperative approach, as is often the case, we consider that the RSUs are not aware of this underlying V2V content-sharing network. The motivation behind this assumption is that, without coordination, the RSUs cannot estimate the fraction of vehicles that can potentially meet and share content, nor the amount of vehicles moving in their direction, and, consequently, they are unaware of the vehicle-to-vehicle content-sharing that can potentially occur. Therefore, in a non-cooperative scenario, it is beneficial for any RSU $i\in \mathcal{N}$ to transmit the chunks of data related to the class $c_1$ with highest priority for all the vehicles passing by this RSU, in order to maximize its revenue\footnote{This is also the most beneficial for the vehicle since, once it enters the network and meets its first RSU, it needs to get the most important class of data.}. Consequently, the non-cooperative utility of any RSU $i\in \mathcal{N}$ can be given by
\begin{equation}\label{eq:ncutil}
u(\{i\}) = \beta\cdot \sum_{\underset{j \neq i}{j \in \mathcal{N}}}\sum_{k \in \mathcal{K}_{ij}} (w_1 \cdot P_{k,i})
\end{equation}
where $w_1P_{k,i}$ represents the effective worth of the data downloaded by vehicle $k$ and $\beta$ is the price charged by the payment operator for one unit of effective data. For example, for an RSU $i$ that is serving the same number of vehicles $K$ in every direction and where each vehicle is downloading the same amount of data $P$, (\ref{eq:ncutil}) becomes $\beta \cdot (N-1) K w_1 P$.

For increasing their revenues, the RSUs can cooperate and exploit the \emph{underlying V2V content-sharing network}. We limit our attention to pairwise V2V content-sharing algorithms such as those in \cite{VAPP} or \cite{V2V01} whereby V2V communication occurs between \emph{pairs} of vehicles that can meet and exchange their data. Further, we consider that this content-sharing occurs prior to the collection of the fees by the operator. Hence, between every pair of RSUs $i \in \mathcal{N}$ and $j\in \mathcal{N},\ i \neq j$, we consider that a certain average number of \emph{pairs of vehicles} can meet and engage in V2V content sharing.  We denote this number by $m_{ij}(d_{ij})$ such that: (i) Each pair of content-sharing vehicles consists from a vehicle moving from $i$ to $j$ which is able to meet up with another vehicle moving from $j$ to $i$, (ii) $m_{ij}(d_{ij})$ is a decreasing function of the distance $d_{ij}$ between RSUs $i$ and $j$ (as the RSUs are more distant, it becomes less likely that the pairs of vehicles circulating between them would meet), and (iii) The maximum number of pairs of vehicles that can meet on a link between an RSU $i$ and an RSU $j$ corresponds to the minimum between $K_{ij}$ and $K_{ji}$, i.e., $m_{ij} \le \min{(K_{ij},K_{ji})}$. The set of vehicles initiating at any RSU $i\in \mathcal{N}$ that will potentially meet other vehicles initiating at any RSU $j\in \mathcal{N}$ over the path between $i$ and $j$ is denoted by $\mathcal{M}_{i}$ (note that, for a path between an RSU $i$ and an RSU $j$, $|\mathcal{M}_i| + |\mathcal{M}_j| = m_{ij}(d_{ij})$). To evaluate the average number of pairs of vehicles $m_{ij}$ that can meet between any two RSUs $i$ and $j$, one can use the following expression:
\begin{equation}\label{eq:conmod}
m_{ij} = \delta^{d_{ij}}\cdot \min{(K_{ij},K_{ji})}
\end{equation}
where $0 \le \delta \le 1$ represents the fraction of vehicles that can meet when the RSUs are distant of $1$~km. The model in (\ref{eq:conmod}) is inspired from the well-known connections model used in network formation games \cite{JA01} for highlighting the friendship relationships among individuals as a function of their distance. Note that, other models for $m_{ij}$ can also be accommodated.

Consequently, given any group of cooperating RSUs, i.e., any coalition $S\subseteq \mathcal{N}$, the RSUs inside $S$ can perform a cooperative protocol composed of the following steps:
\begin{enumerate}
\item Every pair of RSUs $i\in S, j\in S,\ i\neq j$ can communicate over the infrastructure to inform each other of the average number of vehicles $K_{ij}$ that are circulating from RSU $i$ to the direction of RSU $j$ as well as the average number of vehicles $K_{ji}$ moving from RSU $j$ in the direction of RSU $i$. In addition, RSUs $i$ and $j$ can exchange the distance between them and additional information on the behavior of the traffic circulating between them.
    \item Using the exchanged information, every pair of RSUs $i\in S, j\in S,\ \forall i,j \in S\ \textrm{s.t. } i\neq j$ can estimate the number $m_{ij}(d_{ij})$ of pairs of vehicles  that can potentially meet on the path between them.
\item Following the information exchange, the members of coalition $S$ can coordinate the classes of data that each RSU $i \in S$ needs to transmit to its served vehicles. The classes are agreed upon by the coalition members in a way to maximize the amount and diversity of the data received by the vehicles through joint V2R (data downloaded from the RSUs by the vehicles) and V2V (data exchanged between every pair of vehicles that meet) communication. This coordination is limited to selecting the classes of data sent between any two RSUs $i \in S$ and $j \in S$ \emph{inside the same coalition}. For vehicles moving from an RSU $i \in S$ in the direction of RSU $k \in \mathcal{N} \setminus S$ outside $S$, RSU $i$ will still transmit the data of class $c_1$, i.e., highest priority data.
\end{enumerate}

Using this cooperative protocol, the RSUs that are members of any coalition $S$ are able to maximize the revenue they receive from the operator by exploiting the existence of a diversity of data classes as well as the presence of a V2V content-sharing network. Note that, for distinguishing between packets received from RSUs and packets received from other vehicles both RSUs and vehicles can append a small header on each packet in order to indicate its origin. Such an approach is commonly used to distinguish among packets transmitted from different sources \cite{VAPP,PA00}.

Hereafter, we consider that any vehicle moving from an RSU $i \in \mathcal{N}$ to an RSU $j\in \mathcal{N}$ is charged by the network after engaging in V2V communication but right before reaching RSU $j$ (or exiting the network), and, thus, the content that can potentially be downloaded from RSU $j$ is out of the scope of the current model. For instance, the objective of this paper is to highlight the potential benefits that cooperation among the RSUs can yield for the network. Hence, this work emphasizes how exploiting cooperation among the RSUs, even it occurs only during the initial communication stages, can potentially yield important and interesting benefits for the RSUs (more revenues generated by exploiting the underlying V2V network) and the vehicles (more diverse traffic after passing through their first RSU)   as well. Certainly, one can see that coalition formation would still be beneficial when considering the data downloaded by the vehicles at multiple RSUs. For example, if two RSUs cooperate  using our approach while taking into account the data downloaded by the vehicles at \emph{both} RSUs and exploiting the underlying V2V possibilities, they can allow  each vehicle traveling between them to receive three classes of data (one at the first RSU, one through V2V, and one at the second RSU) instead of only two in the non-cooperative case (one at the first RSU and one at the second RSU since the non-cooperative RSUs do not coordinate their data classes based on the underlying V2V network). Thus, coalition formation for the case where the vehicle can download data from multiple RSUs is not considered in this paper but it can be easily accommodated in future work. For example, one approach to do this would be to allow the RSUs to re-engage in the coalition formation algorithm proposed in this paper, periodically, given an estimate of the traveling routes of the vehicles.

Finally, in order to better illustrate the proposed cooperation protocol and model, consider the case of $N=2$ RSUs with an average of $K_{12}=2$~vehicles moving from RSU $1$ in the direction of RSU $2$ and $K_{21}=2$~vehicles moving from RSU $2$ in the direction of RSU $1$. We assume that the total number of packets downloaded by each vehicle from its serving RSU is $1$. Further, we consider that the two RSUs are close enough such that $m_{12}=2$~pairs. In a non-cooperative system, RSU $1$ (RSU $2$) is not aware of the traffic coming from RSU $2$ (RSU $1$) and thus cannot estimate $m_{12}$. Consequently, both RSUs $1$ and RSU $2$ send packets of the highest priority, i.e., $c_1$ and, by setting $w_{c_1}=0.6$, their utilities are  $v(\{1\}) = v(\{2\})=1.2\beta$, respectively, as given by (\ref{eq:ncutil}). By cooperating and forming coalition $S=\{1,2\}$, the two RSUs can increase the revenue they receive. For example, when coalition $S$ forms, the two RSUs can decide that RSU $1$ sends data of class $c_1$ while RSU $2$ sends data of class $c_2$ with weight $w_{c_2} = 0.5$. Due to the V2V content-sharing among the vehicles traveling between the two RSUs (which are members of $S$), prior to reaching the next RSU, the total effective data received by each vehicle is $\beta(w_{c_1} + w_{c_2})$. Consequently, the total revenue generated from an average of $4$ vehicles which are traveling between two RSUs is $4\beta(w_{c_1}+w_{c_2})=4.4\beta$ (since $m_{12}=2$~pairs with $K_{12}=K_{21}=2$, all pairs of vehicles moving in opposite directions between the two RSUs will meet and engage in pairwise V2V content sharing). If the payment operator divides this revenue equally between the members of $S$ and assuming that the RSUs pay no cost for cooperation, then, each RSU receives $2.2\beta$ which is a clear improvement over the non-cooperative utilities.

In a nutshell, by exploiting the underlying V2V content-sharing capabilities the RSUs in a vehicular network can improve the revenues they obtain from  V2R communication with their served vehicles. An illustrative example of a sample network structure is shown in Fig.~\ref{fig:ill} for $N=5$~RSUs and $L=3$ data classes. In this figure, nearby RSUs with relatively high traffic among them form a cooperative coalition while coordinating the classes of data transmitted in a way to exploit the potential V2V content-sharing that can take place.

\begin{figure}[!t]
\begin{center}
\includegraphics[width=80mm]{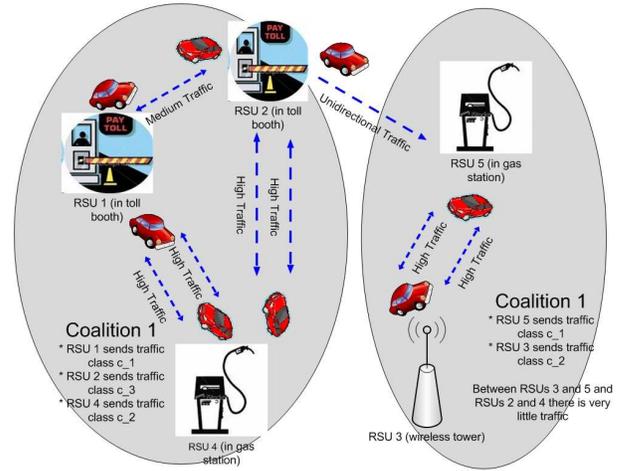}
\end{center}\vspace{-0.6cm}
\caption {An illustrative example of coalition formation for cooperation among RSUs for a network with $N=5$ RSUs and $L=3$ data classes.} \label{fig:ill}
\end{figure}

In the next section, we formulate the cooperation problem among RSUs, develop a suitable utility function, and model the problem as a coalition formation game with transferable utility.

\section{Coalitional Game for Cooperation among RSUs}\label{sec:gmodel}\vspace{-0.1cm}

For mathematically modeling the cooperation problem among the RSUs of a vehicular network, we use coalitional game theory \cite{Game_theory2,Game_theory1,WS00}. In particular, this problem is modeled as a coalitional game with a transferable utility  \cite[Chap. 9]{Game_theory2}:
\begin{definition}
A coalitional game with \emph{transferable utility} is defined by a pair $(\mathcal{N},v)$ where $\mathcal{N}$ is the set of players and $v$ is a function over the real line such that for every coalition $S \subseteq \mathcal{N}$, $v(S)$ is a real number describing the amount of utility that coalition $S$ receives and which can be distributed in any arbitrary manner among the members of $S$.
\end{definition}

For the problem of cooperation among RSUs, given any coalition $S \subseteq \mathcal{N}$, we define $\mathcal{B}_S=\{b_1,\ldots,b_{|S|}\}$ as the tuple with every element $b_i$ representing a class of data in $\mathcal{C}$ selected by an RSU $i \in S$, assumed to be the same for the vehicles starting from RSU $i$ and moving on all directions to RSUs \emph{inside coalition} $S$. Let us denote by $\mathfrak{B}_S$ the family of all such tuples for coalition $S$ which corresponds to the family of all permutations, with repetition, for the RSUs in $S$ over the data classes in $\mathcal{C}$. Note that, as previously mentioned, for the vehicles moving between an RSU $i \in S$ and any RSU $l \in \mathcal{N}\setminus S$ outside coalition $S$, RSU $i$ will always use the class of data with the highest priority, i.e., $c_1$. Consequently, by adopting the cooperative protocol described in the previous section, the total revenue generated by any coalition $S\subseteq \mathcal{N}$ is given by
\begin{align}\label{eq:rev}
u(S) = \!\max_{\mathcal{B}_S \in \mathfrak{B}_S}{\left(\beta\sum_{i\in S}\left(\sum_{j \in \underset{j \neq i}{\mathcal{N}\setminus S}}\sum_{k \in \mathcal{K}_{ij}} (w_{c_1} \cdot P_{k,i}) +\right.\right.}\nonumber \\ {\left.\left.\sum_{\underset{j \neq i}{j \in S}}\sum_{k \in \mathcal{K}_{ij}\setminus\mathcal{M}_{i}}\!\!\!(w_{b_i} \cdot P_{k,i}) + \sum_{\underset{j \neq i}{j \in S}}\sum_{k \in \mathcal{M}_{i}} \left((w_{b_i} + w_{b_j} )\cdot P_{k,i}\right) \right)\right)}
\end{align}
where $w_{c_1}$ represents the weight associated with data class $c_1$ of highest priority and $w_{b_i}$ represents the weight of the data transmitted by RSU $i \in S$ when this RSU selects data class $b_i \in C$ as per tuple $\mathcal{B}_S^{*}$ that maximizes $v(S)$. Further, $\mathcal{M}_i$ represents the set of vehicles starting from RSU $i \in S$ that will engage in V2V content-sharing with the vehicles in the set $\mathcal{M}_j$, i.e., the vehicles initiating at RSU $j\in S$.

The revenue function in (\ref{eq:rev}) represents the total payment that a coalition of RSUs receives from the operator, after the vehicles traveling among the RSUs in $S$  pay the operator prior to reaching their next RSU or, even, exiting the network. For an RSU $i\in S$, the first term in (\ref{eq:rev}) represents the revenue from the data traffic between $i$ and all RSUs \emph{external} to $S$ (that do not rely on any V2V communication). The second term represents the total revenue from the vehicles originating from RSU $i$ and which are unable to meet any vehicle moving from any other RSU $j\in S$. In other words, the second term is the revenue from the data sent to vehicles that travel between two RSUs inside the same coalition without engaging in any V2V content-sharing. Finally, the third term represents the gains from the data received by all vehicles passing by RSU $i \in S$ and which are able to engage in V2V pairwise content-sharing with other vehicles moving from other members of $S$ towards RSU $i$. Hence, the third term in (\ref{eq:rev}) quantifies the gains that the RSUs can receive from exploring the existing V2V communications network\footnote{As previously stated, all V2V communications for content-sharing is considered in a pairwise manner.}.

Note that,  in both (\ref{eq:ncutil}) and (\ref{eq:rev}), we consider that any vehicle $k$ is able to receive successfully the $P_{k,i}$ chunks of data related to a certain class transmitted by RSU $i$. In practice, the actual number of packets successfully received would effectively be less than $P_{k,i}$ due to the wireless channel fading which accounts for the Doppler effect. In order to take this effect into account, in (\ref{eq:ncutil}) and (\ref{eq:rev}), the number of packets received  should be multiplied by a term corresponding to the probability of successful transmission (or probability of error) in presence of channel fading that is a function of the Doppler frequency. Nonetheless, the advantage of cooperation would still hold because whether the effect of channel fading/Doppler is accounted for or not, cooperation will allow the RSUs to exploit the underlying V2V network and, thus, improve the overall performance.


Although cooperation can yield additional revenues for the RSUs as per (\ref{eq:rev}), these gains are limited by inherent costs that need to be paid by the RSUs when acting cooperatively. These costs can be captured by a cost function $c(S)$ which will limit the gains from cooperation obtained in (\ref{eq:rev}). Although the analysis in the remainder of this section can be applied for any type of cost functions, we consider a cost function that varies linearly with the size of the coalition as follows
\begin{equation}\label{eq:lincost}
c(S) = \begin{cases} \alpha  \cdot |S|, & \mbox{if } |S| > 1\\
0, & \mbox{otherwise}\end{cases}
\end{equation}
where $\alpha$ is a pricing factor. The main motivation behind this cost function is that, in order to achieve the revenue function in (\ref{eq:rev}), the RSUs belonging to a single coalition $S$ need to synchronize their communication and maintain an open channel between them to exchange information, determine the classes that need to be sent by each RSU, and so on. Moreover, when the RSUs belong to different service providers, they might be required to pay a certain fee for coordinating their cooperative behavior over each others infrastructure. Consequently, for every coalition $S\subseteq \mathcal{N}$, the RSUs need to pay a cost for coordination which is an increasing function of the coalition size such as in (\ref{eq:lincost}). In fact, for the proposed cooperation protocol one can easily see that, as the number of RSUs in a coalition increases, coordinating the classes for maximizing the revenue becomes more complex and consequently yields additional costs.

Consequently, given the revenue in (\ref{eq:rev}) and the cost in (\ref{eq:lincost}), the net value that any coalition $S$ can receive is given by
\begin{equation}\label{eq:util}
v(S)=u(S) - c(S).
\end{equation}
This value function quantifies the effective revenue that a coalition $S$ of RSUs will receive  given the gains and costs from cooperation. One can easily see that, for a singleton coalition, i.e., $S=\{i\}$, the utility function in (\ref{eq:util}) is reduced to the non-cooperative utility function in (\ref{eq:ncutil}).

For any coalition $S \subseteq \mathcal{N}$, the utility in (\ref{eq:util}) represents the amount of money received by this coalition, and, thus, it can be arbitrarily apportioned among the members of $S$. Therefore, the utility function in (\ref{eq:util}) is considered as a \emph{transferable utility} \cite{Game_theory2}. Consequently, we immediately have the following property:
\begin{property}
The proposed RSUs cooperation problem is modeled as a coalitional game with transferable utility $(\mathcal{N},v)$ where $\mathcal{N}$ is the set of RSUs and $v$ is the value function given by (\ref{eq:util}).
\end{property}

In the proposed coalitional game among RSUs, we denote by $\phi_i(S)$ the \emph{payoff} received by an RSU $i \in S$ after dividing the utility in (\ref{eq:util}) using any payoff division rule and we let $\boldsymbol{\phi}(S)$ be the vector of all payoffs inside coalition $S$.  In this context, the payoff vector $\boldsymbol{\phi}(S)$ is said to be \emph{individually rational} if the
RSUs in $S$ can obtain a benefit in $S$ that is no less than their benefit when acting non-cooperatively, i.e. $\phi_{i}(S) \ge v(\{i\}),\forall i \in S$. In this paper, we adopt an individually rational egalitarian rule for payoff division whereby the \emph{extra} utility (benefit) is divided \emph{equally} among the members of the same coalition. In other words, the payoff $\phi_i(S)$ of any  $i \in S$ is given by
\begin{align}\label{eq:ega}
\phi_i(S)=\frac {1}{|S|} \left( v(S)-\sum_{j\in S} v(\{j\})\right)+v(\{i\})
\end{align}
where $v(\{i\})$ and $v(\{j\})$ are the non-cooperative payoffs of RSU $i$ and RSU $j$. Note that, unlike equal fairness, the egalitarian payoff division rule does not imply dividing the entire utility equally but rather the \emph{extra} benefits equally while conserving individual rationality. Note that, other fairness rules such as the Shapley value or the nucleolus \cite{Game_theory2,Game_theory1,WS00} can also be used.

By carefully inspecting the value function in (\ref{eq:util}), we can state the following remark:
\begin{remark}
In the proposed $(\mathcal{N},v)$ coalitional game, any coalitional structure may form in the network and the grand coalition, i.e., the coalition of all RSUs, is seldom beneficial due to the increasing cooperation costs. Subsequently, the proposed coalitional game among RSUs is classified as a \emph{coalition formation game} \cite{WS00}.
\end{remark}

By carefully inspecting $v(S)$ in (\ref{eq:util}) and through the cost function in (\ref{eq:lincost}) it is clear that as the number of RSUs in a coalition increases, the cost for cooperation increases, and, thus, the gain from cooperation becomes limited by this increasing cost. Consequently, although the grand coalition might form under favorable conditions (e.g., when the cost of cooperation is low relative to the benefits from cooperation), in general, the final network structure is composed of a number of independent disjoint coalitions of RSUs. Hence, the traditional concepts used for solving coalitional games, e.g., the core \cite{Game_theory2}, may not be applicable \cite{WS00}. In brief, the proposed coalitional game among RSUs is an $(\mathcal{N},v)$ \emph{coalition formation game} \cite{WS00} whereby the main goal is to devise an algorithm for allowing the RSUs  to form coalitions such as those shown in  Fig.~\ref{fig:ill} while taking into account both the gain and cost of cooperation.

 \vspace{-0.5cm}

\section{Coalition Formation Algorithm}\label{sec:gform}
In this section, first, we introduce some concepts from coalition formation games, and, then, we devise an algorithm for coalition formation among RSUs in a vehicular network.

\subsection{Coalition Formation Concepts}

As has been mentioned before, the proposed cooperation model entails the formation of disjoint coalitions, and, hence, the proposed game is classified as a coalition formation game. In fact, coalition formation has been a topic of high interest in game theory \cite{CF00,CF01,HC00,HC01,HC02}. One key approach for forming coalitions is to enable the players to join or leave a coalition based on well-defined \emph{preferences}. Such a preference-based approach for coalition formation is the basis of many existing coalition formation concepts such as the merge-and-split algorithm \cite{KA01} or hedonic games \cite{HC00,HC01,HC02}. In this context, we introduce some definitions from coalition formation games, taken from \cite{HC00}.
\begin{definition} A \emph{coalitional structure} or a \emph{coalition partition} is defined as the set $\Pi = \{S_1,\ldots,S_l\}$ which partitions the RSUs' set $\mathcal{N}$, i.e.,  $ \forall\ k\ ,S_k \subseteq \mathcal{N}$ are disjoint coalitions such that $\cup_{k=1}^{l}S_k = \mathcal{N}$ (an example of a partition $\Pi$ composed of $2$ coalitions is shown in Fig.~\ref{fig:ill}).
 \end{definition}
 \begin{definition}
 For any RSU $i\in \mathcal{N}$, given a network partition $\Pi$, we denote by $S_{\Pi}(i)$, the coalition $S_k \in \Pi$, such that $i \in S_k$.
 \end{definition}

For performing preference-based coalition formation in the proposed game, each RSU must build preferences over its own set of possible coalitions. In other words, based on which coalition an RSU prefers to being a member of, the RSU must be able to compare and order its potential coalitions. For evaluating these preferences of the RSUs over the coalitions, we use the concept of a preference relation or order as follows \cite{HC00}:
\begin{definition}
For any RSU $i\in \mathcal{N}$, a \emph{preference relation} or \emph{order} $\succeq_i$ is defined as a complete, reflexive, and transitive binary relation over the set of all coalitions that RSU $i$ can possibly form, i.e., the set $\{S_k \subseteq \mathcal{N} : i \in S_k\}$.
\end{definition}

Hence,  for any given RSU $i\in \mathcal{N}$,  $S_1\succeq_i S_2$, implies that RSU $i$ prefers being a member of coalition $S_1 \subseteq \mathcal{N} $ with $i \in S_1$ over coalition $S_2 \subseteq \mathcal{N} $ with $i \in S_2$,  or at least, $i$ prefers both coalitions equally. The asymmetric counterpart of $\succeq_i$, denoted by $\succ_i$, i.e., $S_1 \succ_i S_2$, implies that RSU $i$ \emph{strictly} prefers being a member of $S_1$ over $S_2$. Every coalition formation application can have a different order for quantifying the players' preferences. This order can be a function of several parameters, such as the payoffs that the players receive from each coalition, the approval of the coalition members, and so on. In this coalition formation game, we propose the following preference relation for any RSU $i\in \mathcal{N}$:

\begin{align}\label{eq:prefsu}
S_1 \succeq_i S_2 \Leftrightarrow r_i(S_1) \ge r_i(S_2)\vspace{-2mm}
\end{align}
where $S_1,\ S_2 \subseteq \mathcal{N}$, are any two coalitions containing RSU $i$, i.e., $i \in S_1$ and $i \in S_2$ and $r_i$ is a preference function defined for any RSU $i\in \mathcal{N}$ and any coalition $S$ such that $i\in S$ as follows:
\begin{align}\label{eq:pref1}
r_i(S) = \begin{cases} \phi_i(S), & \mbox{if } \phi_j(S) \ge \phi_j(S \setminus \{i\}),\forall j \in S\setminus\{i\} \\& \&\ S \notin h(i)  \mbox { or } (|S|=1)  \\ -\infty, &\mbox{otherwise} \end{cases}
\end{align}
where $\phi_i(S)$ is the payoff  received by RSU $i$ in coalition $S$ as per the egalitarian fair rule given by (\ref{eq:ega}) and $h(i)$ is a history set where RSU $i$ stores the identity of the coalitions that it visited and then left in the past. Note that $h(i)$ is only applicable to coalitions with size larger than $1$ since it is natural to consider that, at any time, any RSU can always revert to acting non-cooperatively.

Having clearly defined the required coalition formation concepts, the next step is to provide a distributed algorithm, based on the defined preferences, for forming the coalitions.

\subsection{Distributed Coalition Formation Algorithm}
For constructing a coalition formation process, we introduce an algorithm that allows the RSUs to take distributed decisions for selecting which coalitions to join at any point in time. In this regard, we propose the following rule for coalition formation:
\begin{definition}\label{def:switch}
\textbf{Switch Operation -} Given a partition $\Pi=\{S_1,\ldots,S_l\}$ of the set $\mathcal{N}$, any RSU $i \in \mathcal{N}$ decides to leave its current coalition $S_{\Pi}(i)=S_m,\ $ for some $m \in \{1,\ldots,l\}$ and join another coalition $S_k \in \Pi \cup \{\emptyset\},\ S_k \neq S_{\Pi}(i)$, if and only if $S_k \cup \{i\} \succ_i S_{\Pi}(i)$. Hence, $\{S_m,S_k\} \rightarrow \{S_m\setminus\{i\},S_k\cup\{i\}\}$.
\end{definition}

Hence, for every single switch operation made by an RSU $i$, a current partition $\Pi$ of $\mathcal{N}$ is modified into a new partition $\Pi^{\prime}$ such that $\Pi^{\prime} = (\Pi \setminus \{S_m,S_k\}) \cup \{S_m\setminus\{i\},S_k\cup\{i\}\}$. Therefore, for every partition $\Pi$, the switch operation constitutes a mechanism using which any RSU can leave its current coalition $S_{\Pi}(i)$, and join another coalition $S_k \in \Pi$, given that the new coalition $S_k \cup \{i\}$ is strictly preferred over $S_{\Pi}(i)$ through the preference relation defined in (\ref{eq:pref1}). In fact, by carefully inspecting (\ref{eq:pref1}), we note that, an RSU $i$ would perform a switch operation by leaving  $S_{\Pi}(i)$ and joining  $S_k \in \Pi$ if this RSU can \emph{strictly} improve its payoff by joining the new coalition $S_k \cup \{i\}$ \emph{without decreasing} any of the payoffs of the RSUs in coalition $S_k$ and given that RSU $i$ did not join and leave $S_k \cup \{i\}$ in the past. Hence, the switch operation using (\ref{eq:pref1}) is a mechanism through which an RSU decides to join a new coalition given the \emph{consent} of all the members of the coalition. Furthermore, we consider that, whenever an RSU decides to switch from one coalition to another, it updates its history set $h(i)$. Hence, given a partition $\Pi$, whenever an RSU $i$ decides to leave coalition $S_l \in \Pi$ to join another coalition, coalition $S_l$ is stored by RSU $i$ in its history set $h(i)$.

Using the switch operation, we construct a coalition formation algorithm composed of three main stages: Neighbor discovery, distributed coalition formation, and cooperative V2R communication. In the first stage, each RSU uses a neighbor discovery algorithm in order to identify potential candidate RSUs (or coalitions) for cooperation. For discovering their neighbors, the RSUs can utilize several well-known neighbor discovery algorithms (e.g., those used in ad hoc routing discovery \cite{ADHOC00} or wireless networks \cite{ZH00}). The outcome of the neighbor discovery stage is that every RSU becomes aware of its neighbors as well as of the current network structure. Following neighbor discovery, the RSUs can interact with each other for performing cooperation using coalition formation. Thus, the second stage of the algorithm is the coalition formation stage whereby, all the RSUs investigate the possibility of performing a switch operation. In this context, every RSU  investigates its top preferred coalition and decides to perform a switch operation, if possible, through (\ref{eq:pref1}). We assume that the order in which the RSUs make their switch operations is random. For any RSU, a switch operation is easily performed as the RSU can change its coalition membership by leaving its current coalition and joining the new coalition, given the approval of the RSUs in the new coalition as per (\ref{eq:pref1}). The convergence of the proposed coalition formation algorithm during this second stage is guaranteed as follows:
\begin{theorem}\label{th:one}
Starting from any initial coalitional structure $\Pi_{\textrm{initial}}$, the coalition formation stage of the proposed algorithm maps to a sequence of switch operations which will always converge to a final network partition $\Pi_f$ composed of a number of disjoint coalitions.
\end{theorem}
\begin{proof}
For the purpose of this proof, we denote by $\Pi_{n_{k}}^{k}$ the partition formed during the turn $k$ of any RSU $i \in \mathcal{N}$ after the occurrence of $n_{k}$ switch operations (the index $n_{k}$ denotes the number of switch operations performed by the RSUs that made their coalition decisions up to the turn of RSU $i$). Beginning with any initial starting partition $\Pi_{\textrm{initial}}=\Pi^{1}_{0}$, the coalition formation phase of the proposed algorithm can be easily mapped to a sequence of switch operations. As per definition~\ref{def:switch}, every switch operation transforms the current partition $\Pi$ into another partition $\Pi^\prime$, hence, the sequence of switch operations yields the following transformations (as an illustrative example):
\begin{equation}\label{eq:trans}
\Pi^{1}_{0}=\Pi^{2}_{0}\rightarrow \Pi^{3}_{1} \rightarrow \ldots \rightarrow \Pi^{T}_{n_{T}}
\end{equation}
where the operator $\rightarrow$ indicates the occurrence of a switch operation and $T$ indicates the total number of turns taken by the RSUs. In other words, $\Pi_{n_{k}}^{k} \rightarrow \Pi_{n_{k+1}}^{k+1}$, means that during turn number $k$, an RSU $i$ made a single switch operation which yielded a new partition $\Pi_{n_{k+1}}^{k+1}$ at turn $k+1$. Note that, in (\ref{eq:trans}), the first element $\Pi^{1}_{0}=\Pi^{2}_{0}$ implies that, at the turn of RSU $1$, \emph{no} switch operations occurrs (indicated by  subscript $0$) and, thus, the second partition $\Pi^{2}_{0}$ is still equal to the first partition $\Pi^{1}_{0}$. By inspecting the preference relation defined in (\ref{eq:pref1}), it is easily seen that every single switch operation leads to a partition that has not yet been visited (new partition) or yields a previously visited partition with a non-cooperative RSU (a coalition of size $1$).

In the case where every single switch operations leads to a previously unvisited coalition, and due to the well known fact that the number of partitions of a set is \emph{finite} and given by the Bell number \cite{CF00}, the number of transformations in (\ref{eq:trans}) is finite, and hence the sequence in (\ref{eq:trans}) will always terminate and converge to a final partition $\Pi_f = \Pi^{T}_{n_{T}}$ after $T$ turns. In the case where a previously visited partition with a non-cooperative RSU is visited, starting from this partition, at a certain point in time, the non-cooperative RSU must either join a new coalition and, thus, as per (\ref{eq:pref1}) yield an unvisited partition or decide to remain non-cooperative. Subsequently, the number of re-visited partitions will be limited, and, thus, on all cases, the coalition formation stage of the proposed algorithm will converge to a final network partition $\Pi_f$ composed of a number of disjoint RSUs coalitions, which completes the proof.
\end{proof}

The stability of the final partition $\Pi_f$ resulting from the convergence of the proposed algorithm can be studied using the following stability concept \cite{HC00}:
\begin{definition}
A partition $\Pi = \{S_1,\ldots,S_l\}$ is  \emph{Nash-stable} if $\forall i \in \mathcal{N},\  S_{\Pi}(i) \succeq_i S_k \cup \{i\}$ for all $S_k \in \Pi \cup \{\emptyset\}$.
\end{definition}

This definition implies that, any coalition partition $\Pi$ where no RSU has an incentive to move from its current coalition to another coalition in $\Pi$ or to deviate and act alone is considered as a Nash-stable partition. Moreover, a Nash-stable partition $\Pi$ implies that there does not exist any coalition $S_k \in \mathcal{N}$ such that an RSU $i$ strictly prefers, as per (\ref{eq:pref1}), to be part of $S_k$ over being part of its current coalition (given that the RSUs in $S_k$ do not get hurt by forming $S_k \cup \{i\}$). This is the concept of individual stability, which is formally defined as follows \cite{HC00}:
\begin{definition}
A partition $\Pi = \{S_1,\ldots,S_l\}$ is \emph{individually stable} if there do not exist $i \in \mathcal{N},$ and a coalition $S_k \in \Pi \cup \{\emptyset\}$ such that $  S_k \cup \{i\} \succ_i S_{\Pi}(i)$ and $  S_k \cup \{i\} \succeq_j S_k$ for all $j \in S_k$.
\end{definition}
As already noted, a Nash-stable partition is individually stable \cite{HC00} and, thus, we have the following:
\begin{proposition}\label{prop:one}
Any partition $\Pi_f$ resulting from the coalition formation phase of the proposed algorithm is Nash-stable, and, hence, individually stable.
\end{proposition}
\begin{proof}
 Assume that the partition $\Pi_f$ resulting from the proposed algorithm is not Nash-stable. Consequently, there exists an RSU $i \in \mathcal{N}$, and a coalition $S_k \in \Pi_f$ such that $S_k \cup \{i\} \succ_i S_{\Pi_f}(i)$, hence, RSU $i$ can perform a \emph{switch} operation which contradicts with the fact that $\Pi_f$ is the result of the convergence of the proposed algorithm (Theorem~\ref{th:one}). Consequently, any partition $\Pi_f$ resulting from the coalition formation phase of the proposed algorithm is Nash-stable, and, hence, by \cite{HC00}, this resulting partition is also individually stable.
\end{proof}
 \begin{table}[!t]
  \centering
  \caption{
    \vspace*{-0.2em}The proposed coalition formation algorithm for roadside units~(RSUs) cooperation.}\vspace*{-1em}
    \begin{tabular}{p{8cm}}
      \hline
      \textbf{Starting Network} \vspace*{.5em} \\
      \hspace*{1em}The RSUs in the network are organized into an initial partition \vspace*{.4em}\\ \hspace*{1em}$\Pi_{\textrm{initial}}=\{S_1,\ldots,S_k\}$. At the beginning  $\Pi_{\textrm{initial}}$ = $\mathcal{N}$. \vspace*{.4em}\\
\textbf{Coalition Formation Algorithm with Three Stages} \vspace*{.1em}\\
\hspace*{1em}\emph{Stage I - Neighbor Discovery:}   \vspace*{.1em}\\
\hspace*{1.5em}Each RSU discovers neighboring RSUs as well as the current\vspace*{.1em}\\
\hspace*{1.5em}network structure $\Pi_{\textrm{initial}}$.\vspace*{.1em}\\
\hspace*{1em}\emph{Stage II -  Coalition Formation:}\vspace*{.1em}\\
\hspace*{2em}In this stage, the RSUs engage in coalition formation as follows:\vspace*{.2em}\\
\hspace*{3em}\textbf{repeat}\vspace*{.2em}\\
\hspace*{3em}For every RSU $i \in \mathcal{N}$, given any current partition $\Pi_{\textrm{current}}$\vspace*{.2em} (in\vspace*{.2em}\\
\hspace*{3em}the first round $\Pi_{\textrm{current}} = \Pi_{\textrm{initial}}$).\vspace*{.2em}\\
\hspace*{4em}a) RSU $i$ investigates possible switch operations using the\vspace*{.2em}\\
\hspace*{4em}preferences given by (\ref{eq:pref1})\vspace*{.2em}\\
\hspace*{4em}b) If possible, RSU $i$ performs a switch operation as follows:\vspace*{.2em}\\
\hspace*{5em}b.1) RSU $i$ updates its history $h(i)$ by adding coalition\vspace*{.2em}\\ \hspace*{5em}$S_{\Pi_{\textrm{current}}}(i)$, before leaving it.\vspace*{.2em}\\
\hspace*{5em}b.2) RSU $i$ leaves its current coalition $S_{\Pi_{\textrm{current}}}(i)$.\vspace*{.2em}\\
\hspace*{5em}b.3) RSU $i$ joins the new coalition that improves its payoff.\vspace*{.2em}\\
\hspace*{3em}\textbf{until} convergence to a final Nash-stable partition $\Pi_{f}$.\vspace*{.2em}\\
\hspace*{1em}\emph{Stage III - Cooperative V2R Communication}   \vspace*{.1em}\\
\hspace*{2em}a) The current network is partitioned using $\Pi_{\textrm{f}}$.\vspace*{.1em}\\
\hspace*{2em}b) The coalitions of RSUs operate using the cooperative protocol\vspace*{.1em}\\
\hspace*{2em}discussed in Section~\ref{sec:sysmodel}.\vspace*{.1em}\\
\hspace*{1em}\emph{Adaptation to environmental changes (periodic process)}\vspace*{.1em}\\
\hspace*{2em}In the presence of environmental changes, such as a change in\vspace*{.1em}\\
\hspace*{2em}the vehicle traffic, every period of time $\Psi$, the three stages of the\vspace*{.1em}\\
\hspace*{2em}algorithm are repeated to allow the RSUs to self-organize and\vspace*{.1em}\\
\hspace*{2em}adapt the network structure to these environmental changes.\vspace*{.1em}\\

   \hline
    \end{tabular}\label{tab:alg}
\end{table}

Following the formation of the coalitions and the convergence of the coalition formation phase to a Nash-stable partition, the last phase of the algorithm entails the actual cooperative V2R communication within every formed coalition as explained in Sections~\ref{sec:sysmodel} and \ref{sec:gmodel}. The proposed algorithm is summarized in Table~\ref{tab:alg}.

\subsection{Adaptation to Environmental Changes and Implementation}

Using the algorithm proposed in Table~\ref{tab:alg}, the RSUs can adapt the network structure to environmental changes. One important environmental change that can occur in V2R networks is a change in the traffic passing by every RSU. Such a change automatically modifies the meeting possibilities among the vehicles traveling between the RSUs  and, thus, there is a need for the network to re-organize. For this purpose, the three stages of the algorithm shown in Table~\ref{tab:alg} are repeated periodically over time to adapt to any changes that have occurred in the environment. As per Theorem~\ref{th:one} and Proposition~\ref{prop:one}, regardless of the starting position, the players will always self-organize into a Nash-stable partition, even after any environmental change. In brief, in a changing environment, the proposed algorithm in Table~\ref{tab:alg} is repeated every period of time $\Psi$. This period of time is chosen depending on how rapidly the environment is changing, e.g., for rapidly changing environments $\Psi$ is chosen to have a small value while for static environments it can have a large value. Note that, every $\Psi$, the history set $h(i)$ for all RSUs $i\in \mathcal{N}$ is reset.

To implement the proposed algorithm, a distributed approach can be used. For instance, every switch operation can be taken by the RSUs individually without relying on any centralized entity. For discovering their neighbors,  the RSUs can utilize well-known algorithms such as those discussed in \cite{ADHOC00} or \cite{ZH00}. Further, the RSUs are required to evaluate their potential payoff as per (\ref{eq:ega}) in order to make an accurate switch operation. For evaluating this payoff, the RSUs can communicate over the underlying network infrastructure which, in general, can enable a reliable communication \cite{VAPP}. By doing so, the RSUs can exchange the needed information such as their estimates on the average vehicles passing in each direction, their location as well as the permutation of classes that can maximize the utility in (\ref{eq:util}). Although at first glance, finding the permutation of classes that maximizes (\ref{eq:util}) seem to be complex, this complexity is reduced due to the fact that both the number of data classes in the network and the sizes of the formed coalition sizes (due to cooperation cost) are generally small. Further, given a present partition $\Pi$, for every RSU, the computational complexity of finding its next coalition, i.e., performing a switch operation, is easily seen to be $O(|\Pi|)$, and in the worst case (i.e., when all the players are non-cooperative), $|\Pi| = N$.   Finally, in changing environments, as the algorithm is repeated periodically, the complexity of the coalition formation algorithm is comparable to the one in the static environment, but with more runs of the algorithm.

\section{Simulation Results and Discussions}\label{sec:sim}

We set up the following network for simulations: We consider a $3$~km$\times 3$~km square area within which the RSUs are randomly deployed. The total number of data classes in the network is set to $L=3$~classes with the corresponding set being $\mathcal{C}=\{c_1,c_2,c_3\}$ and the weights given by $w_{c_1} = 0.9,\ w_{c_2} = 0.8,\ $ and $w_{c_3}=0.7$. The total number of vehicles $K_{ij}$ initiating from any RSU $i\in \mathcal{N}$ in the direction of any other RSU $j \in \mathcal{N}$ is considered equal at all directions, i.e., $K_{ij}=K_i\ \forall j \in \mathcal{N}$ s.t. $j \neq i$. Then, this number of vehicles $K_i$ for an RSU $i$ is selected from a uniform distribution over the integers  such that $K_{i} \le 25, \forall i \in \mathcal{N}$. The average number of chunks of data downloaded by any vehicle $k$ from any RSU $i \in \mathcal{N}$ is set to $P_{k,i}=10\ \forall i,k$.  The price per effective data unit is set to $\beta=1$ while the price factor for the cost is set to $\alpha =10$.  Unless stated otherwise, the fraction of vehicles that can meet when the RSUs are distant of $1$~km is set to $\delta=0.8$. All of the statistical results presented in this section are averaged over the random positions of the RSUs as well as the random traffic pattern (random average number of vehicles $K_i$ initiating from any RSU $i\in \mathcal{N}$ within the above-mentioned range).

\begin{figure}[!t]
\begin{center}
\includegraphics[width=8cm]{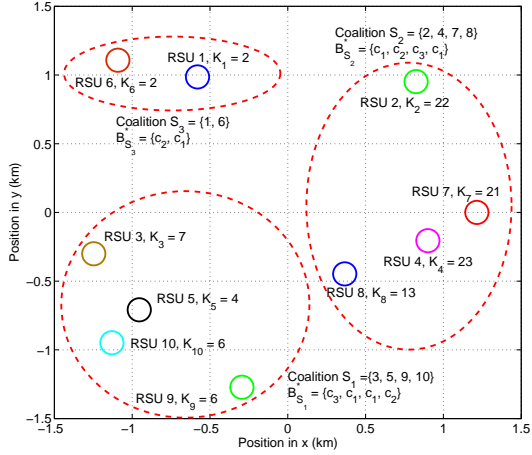}
\end{center}\vspace{-0.8cm}
\caption {A snapshot of the coalitional structure $\Pi_f=\{S_1,S_2,S_3\}$ resulting from the proposed algorithm for a network with $N=10$~RSUs, $\delta=0.8$, and different number of vehicles initiating at every RSU (in this figure, the number of vehicles starting at RSU $i$ is denoted by $K_i$ and assumed the same for all directions inside the coalition). For each coalition $S$, the utility-maximizing data classes permutation tuple $\mathcal{B}_S^*$ is shown.} \label{fig:snap}
\end{figure}

In Fig.~\ref{fig:snap}, we show a snapshot of the network structure $\Pi_f = \{S_1,S_2,S_3\}$ resulting from the proposed algorithm for a vehicular network with $N=10$ randomly deployed RSUs.  Fig.~\ref{fig:snap} shows $3$ different coalitions which are formed among the different RSUs as a result of the proposed coalition formation algorithm. In this figure, for every RSU $i$ belonging to a coalition $S_k \in \Pi_f$, we show the average total number of vehicles $K_i$ starting at $i$. In addition, for every coalition $S_k \in \Pi_f$, we show the corresponding utility-maximizing tuple $\mathcal{B}_{S_k}^*$ of data classes obtained through coordination among the cooperating RSUs. First, we notice that that all three formed coalitions encompass groups of nearby RSUs with comparable vehicle traffic. This corroborates the natural result that RSUs that have an almost similar traffic and that are closely located are more apt to form a coalition. Further, the partition in Fig.~\ref{fig:snap} is Nash-stable since no RSU has an incentive to switch its current coalition. For example, RSU $9$ has a payoff of $\phi_9(S_1)=553.5$ when being part of coalition $S_1=\{3,5,9,10\}$, and by switching to act non-cooperatively, this payoff drops to $\phi_9(\{9\})=486$. Similarly, if RSU $9$ wants to switch from $S_1$ to join with coalition $S_3 =\{1,6\}$, its payoff drops to $\phi_9(\{1,6,9\})=281.4$.

Finally, although by joining with coalition $S_2=\{2,4,7,8\}$ RSU $9$ can significantly improve its utility to $\phi_9(\{2,4,7,8,9\})=1604.2$,  the members of $S_2$ do not agree on the joining of RSU $9$ since this would decrease the payoffs of RSUs $2$, $4$, and $7$ which  drop from $\phi_2(S_2)= 2037.3$, $\phi_4(S_2)=2118.3$, and $\phi_7(S_2)=1956.3$ to $\phi_2(\{2,4,7,8,9\})= 1597.2$, $\phi_4(\{2,4,7,8,9\})= 1599.2$, and $\phi_7(\{2,4,7,8,9\})=1602.1$, respectively. This drop is a consequence of the cost of cooperation, the distance of RSU $9$ to the members of coalition $S_2$, as well as the low possibility of V2V content-sharing on the links between RSU $9$ and the members of $S_2$ (due to the relatively lower average number of vehicles at RSU $9$ compared to that at the other members of $S_2$). In a nutshell, Fig.~\ref{fig:snap} provides an insight on how the RSUs can self-organize using the proposed coalition formation algorithm.

\begin{figure}[!t]
\begin{center}
\includegraphics[width=8cm]{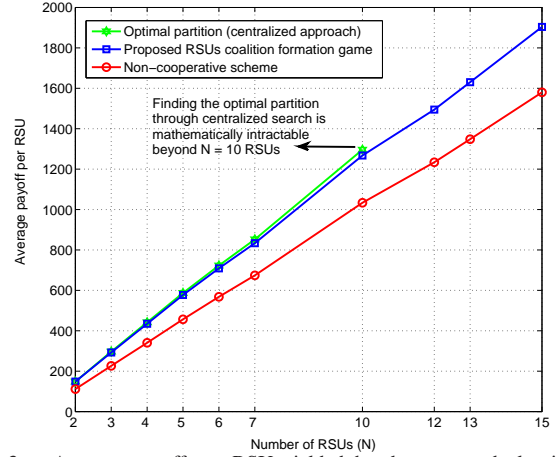}
\end{center}\vspace{-0.8cm}\caption {Average payoff per RSU yielded by the proposed algorithm, the non-cooperative scheme, and the optimal partition (centralized approach) as a function of the number of RSUs $N$ for networks with $L=3$~data classes and $\delta=0.8$.} \label{fig:perf}
\end{figure}

In Fig.~\ref{fig:perf}, we show the average payoff (averaged over all RSU positions and all vehicle traffic patterns) achieved per RSU for a network with $L=3$~data classes and $\delta=0.8$ as the number of RSUs in the network, $N$, increases. In this figure, we compare the performance of the proposed algorithm to that of the non-cooperative case as well as the optimal partition, i.e., the partition that maximizes the average payoff per RSU, found by a centralized entity through exhaustive search. Fig.~\ref{fig:perf} shows that, as the number of RSUs $N$ increases, the performance of the proposed scheme, the non-cooperative scheme, and the optimal partition, increases. The increase in the performance of the non-cooperative scheme with the network size is \emph{solely} due to the existence of additional data traffic yielded by the additional RSUs. For the proposed algorithm and the optimal partition, the increase in the average payoff per RSU with the network size $N$ is also a result of the increased possibility of finding better cooperating partners as the network grows. At all network sizes, the proposed coalition formation algorithm maintains a performance advantage compared to the non-cooperative case. This advantage ranges between $20.5\%$ and $33.2\%$ of improvement, respectively, at $N=15$ and $N=2$  relative to the average non-cooperative payoff.

Moreover, compared to the optimal solution, clearly the proposed coalition formation algorithm achieves a highly comparable performance with a performance gap not exceeding $2.3\%$ with respect to the optimal solution at $N=10$~RSUs. This shows that, by using the proposed distributed coalition formation algorithm, the RSUs can achieve a performance that is very close to optimal. Note that, for more than $10$~RSUs, finding the optimal partition by exhaustive search is mathematically and computationally intractable. Finally, Fig.~\ref{fig:perf} also shows that, as more RSUs are deployed in the network, more traffic is served and, thus, the average payoff per RSU increases significantly for all three schemes. Fig.~\ref{fig:perf} shows that adding $5$~RSUs increases the average payoff per RSU of around $400\%$ to $500\%$ for both the non-cooperative and the cooperative cases which provides incentives for deploying additional RSUs and, thus, generating more revenues per RSU.

\begin{figure}[!t]
\begin{center}
\includegraphics[width=8cm]{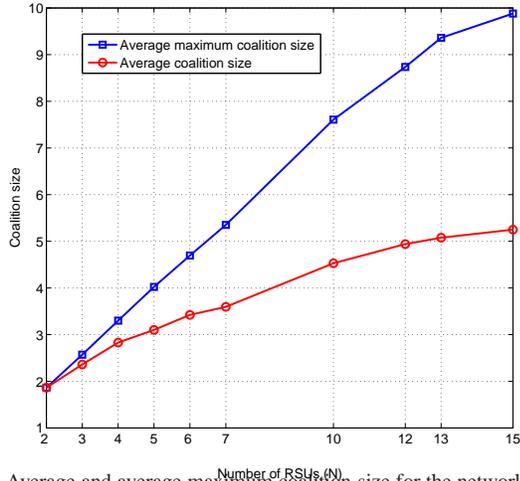}
\end{center}\vspace{-0.8cm}\caption {Average and average maximum coalition size for the networks yielded by the proposed coalition formation algorithm as a function of the number of RSUs $N$ for networks with $L=3$~data classes and $\delta=0.8$.} \label{fig:size}
\end{figure}

In Fig.~\ref{fig:size}, we show the average and average maximum coalition size (averaged over the random positions of the RSUs and the random vehicle traffic pattern) resulting from the proposed  coalition formation algorithm as the number of RSUs, $N$, increases, for a network with $L=3$~data classes and $\delta=0.8$. Fig.~\ref{fig:size} shows that  both the average and the average maximum coalition size increase with the number of RSUs. This is mainly due to the fact that, as $N$ increases, the number of candidate cooperating partners increases, thus, increasing the average size of the formed coalitions. Fig.~\ref{fig:size} also shows that the formed coalitions have a moderate to large size, with the average and average maximum coalition size ranging, respectively, from around $1.9$ (for both) at $N=2$ to around $5.25$ and $9.9$ at $N=15$. Therefore, we can conclude that, in general,  after coalition formation, the resulting network of RSUs is mainly composed of a small number of relatively large coalitions rather than a large number of small coalitions.

\begin{figure}[!t]
\begin{center}
\includegraphics[width=8cm]{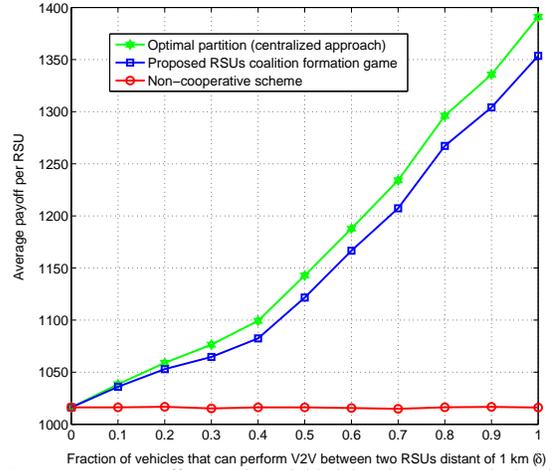}
\end{center}\vspace{-0.8cm}\caption {Average payoff per RSU yielded by the proposed algorithm, the optimal scheme, and the non-cooperative scheme as a function of the fraction $\delta$ of vehicles that can meet and engage in V2V content-sharing between any two RSUs distant $1$~km for a network with $N=10$~RSUs and $L=3$~data classes.} \label{fig:perfpar}
\end{figure}

In Fig.~\ref{fig:perfpar}, we show the average payoff achieved per RSU for a network with $N=10$~RSUs and $L=3$~data classes as the parameter $\delta$ varies. The parameter $\delta$, as given in (\ref{eq:conmod}), represents the fraction of vehicles that can meet and engage in V2V content sharing between any two RSUs distant of $1$~km. Fig.~\ref{fig:perfpar} shows that, as $\delta$ increases, the performance of the proposed scheme as well as of the optimal solution increase while that of the non-cooperative scheme remains constant at all $\delta$. In the non-cooperative scheme, due to the fact that the RSUs are unaware of the underlying V2V communications network, the performance of the non-cooperative approach is not affected by the variations in $\delta$. In contrast, for the proposed coalition formation algorithm and the optimal solution, as $\delta$ increases, it becomes more beneficial for the RSUs to exploit the V2V communications network, and, thus, the performance gains in terms of average payoff per RSU increases with $\delta$. In this context, the proposed coalition formation algorithm presents a performance advantage over the non-cooperative scheme at all $\delta$. This advantage is increasing with $\delta$ and it reaches up to around $25\%$ of improvement relative to the non-cooperative approach at $\delta=1$. Moreover, Fig.~\ref{fig:perfpar} shows that, at all $\delta$ the proposed coalition formation algorithm yields a near optimal performance as the performance gap with respect to the optimal solution does not exceed $2.8\%$ (achieved at $\delta=1$). Finally, as can be seen in Fig.~\ref{fig:perfpar}, at $\delta=0$, there is practically no possibility that V2V communications take place, and, hence, at this value the coalition formation scheme reduces to a non-cooperative approach.

\begin{figure}[!t]
\begin{center}
\includegraphics[width=8cm]{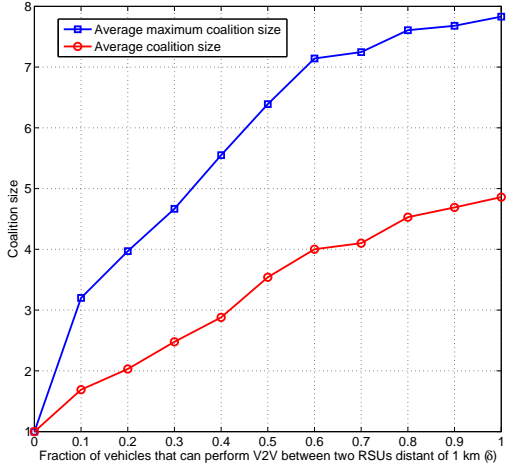}
\end{center}\vspace{-0.8cm}\caption {Average and average maximum coalition size for the networks yielded by the proposed coalition formation algorithm as a function of the fraction $\delta$ of vehicles that can meet and engage in V2V content-sharing between any two RSUs distant of $1$~km for a network with $N=10$~RSUs and $L=3$~data classes.} \label{fig:sizepar}
\end{figure}

In Fig.~\ref{fig:sizepar}, we show the average and average maximum coalition size resulting from the proposed  coalition formation algorithm for a network with $N=10$~RSUs and $L=3$~data classes as the parameter $\delta$ varies. Fig.~\ref{fig:sizepar} shows that  both the average and average maximum coalition size increase with $\delta$. This is mainly due to the fact that, as $\delta$ increases, the amount of V2V data exchange increases and, hence, cooperation becomes highly desirable. From Fig.~\ref{fig:size}, we note that, when $\delta \le 0.4$, the network structure tends towards a partition composed of large number of small coalitions with the average and average maximum coalition sizes not exceeding $3$ and $5.6$, respectively, at $\delta=0.4$. However, as $\delta$ increases beyond $0.4$, the emergence of large coalitions becomes more likely as the average and average coalition size reach up to $4.9$ and $7.8$, respectively, at $\delta=1$.

\begin{figure}[!t]
\begin{center}
\includegraphics[width=8cm]{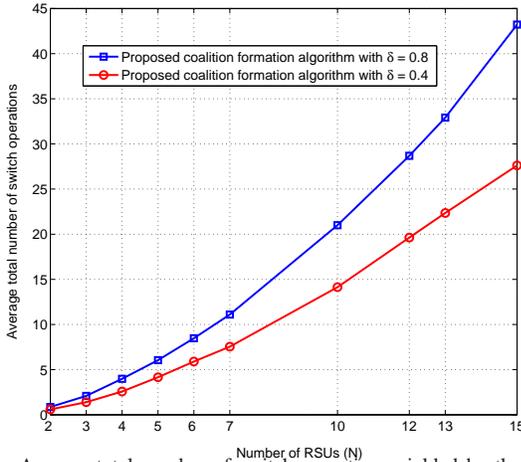}
\end{center}\vspace{-0.8cm}\caption {Average total number of switch operations yielded by the proposed coalition formation algorithm prior to convergence as a function of the number of RSUs $N$ for networks with $L=3$~data classes and for the cases of $\delta=0.4$ and $\delta=0.8$.} \label{fig:switch}
\end{figure}

In Fig.~\ref{fig:switch}, we show the average total number of switch operations  (averaged over the random positions of the RSUs and the random vehicle traffic pattern) that occur prior to the convergence of the proposed algorithm as the number of RSUs, $N$, increases, for a network with $L=3$~data classes and for the cases of $\delta=0.4$ and $\delta=0.8$. Fig.~\ref{fig:switch} shows that, for both $\delta=0.4$ and $\delta=0.8$, the average total number of switch operations increases with the network size. This is mainly due to the fact that, as $N$ increases, the possibilities for cooperation increase, yielding an increased number of switch operations. In this figure, we remark the the total number of switch operations required for the convergence of the coalition formation algorithm varies from $0.6$ and $0.87$ at $N=2$ to around $27.6$ and $43.2$ at $N=15$~RSUs for the cases of $\delta=0.4$ and $\delta=0.8$, respectively. This result implies that for a network of $N=15$~RSUs with $\delta=0.8$ an average of $3$ switch operations per RSU are required before convergence, which demonstrates that the convergence time of the proposed algorithm is quite reasonable. Finally, in Fig.~\ref{fig:switch}, we note that as the possibilities of V2V content-sharing decreases, i.e., as $\delta$ decreases, the total number of switch operations decreases. For instance, at all network sizes, the total number of switch operations required for convergence at $\delta=0.4$ is approximately two-third of that required for the convergence of the algorithm at $\delta=0.8$.

\begin{figure}[!t]
\begin{center}
\includegraphics[width=8cm]{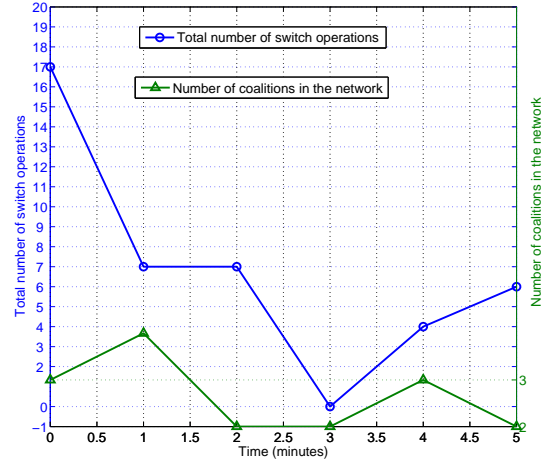}
\end{center}\vspace{-0.8cm}\caption {Evolution of the network structure over time as the vehicle traffic pattern (the average number of vehicles $K_i$ served by each RSU) varies over a period of $5$~minutes for a network with $N=10$~RSUs.} \label{fig:mobconv}
\end{figure}

Fig.~\ref{fig:mobconv} shows how the structure of a vehicular network with $N=10$~RSUs evolves and self-adapts over time over a period of $5$ minutes, while the average traffic $K_i$ present at each RSU $i\in\mathcal{N}$ varies every $1$~minute. The proposed coalition algorithm is repeated periodically by the RSUs every $\Psi=1$ minute, in order to provide self-adaptation to the change in the traffic pattern. The network starts with a non-cooperative structure made up of $10$ independent RSUs. In the first step, as shown in Fig.~\ref{fig:mobconv} at $t=0$, the network self-organizes into $3$ coalitions through a total of around $17$ switch operations. As time evolves, the RSUs perform various switch operations as the network topology changes with the emergence of new coalitions and the departure of other coalitions. For example, after $1$~minute, the RSUs perform a total of $7$ switch operations as the network structure changes from a partition composed of $3$ coalitions at $t=0$ to a partition composed of $4$ coalitions at $t=1$. Between $t=2$~minutes and $t=3$~minutes the network structure remains unchanged as no switch operations occur. Finally, once all $5$ minutes have elapsed the network structure is made up of $2$ coalitions after a total of $41$ switch operations have occurred since $t=0$.

\section{Conclusion}\label{sec:conc}

We have introduced a novel model for distributed cooperation among the roadside units in a vehicular network. In the proposed model, any group of cooperating roadside units can coordinate the classes of data they transmit to the vehicles, and, thus, improve the diversity of the data circulating in the network. In addition, using the proposed cooperation model, any group of cooperating roadside units can improve their revenue by exploiting the underlying vehicle-to-vehicle content-sharing network. The proposed  cooperation model for roadside units has been formulated as a coalition formation game with transferable utility and a coalition formation algorithm has been proposed. Using the proposed coalition formation algorithm, the roadside units can take individual  decisions to join or leave a coalition while maximizing their payoffs. The payoff accounts for the gains from cooperation, in terms of increased revenues as well as the cost of coordination. We have studied the properties and characteristics of the proposed model and showed that the proposed coalition formation algorithm always converges to a Nash-stable partition. Further, by  repeating the proposed coalition formation algorithm periodically, the roadside units can take autonomous decisions for adapting the network structure to environmental changes such as a change in the vehicle traffic. Simulation results show how the proposed algorithm allows the roadside units to self-organize into independent coalitions, while improving the performance, in terms of the average payoff per roadside unit between $20.5\%$ and $33.2\%$ (depending on different scenarios) relative to the non-cooperative scheme.

\nocite{WS00}
\renewcommand{\baselinestretch}{1}
\bibliographystyle{IEEEtran}
\bibliography{references}

\end{document}